\documentclass[letterpaper, 11pt]{article}

\usepackage{appendix}

\newcommand{\comment}[1]{}

\usepackage{graphicx}
\usepackage{latexsym}
\usepackage{mdwlist}

\usepackage{algorithm}

\newcommand{\BSB}{{\em Broadcast\_Single\_Bit~}}
\newcommand{\DG}{{\em Diag\_Graph~}}
\newcommand{\TRUE}{{\bf true~}}
\newcommand{\FALSE}{{\bf false~}}

\newenvironment{proof}{\paragraph{\bf Proof:}}{\hspace*{\fill}\(\Box\)}
\newtheorem{theorem}{Theorem}

\newtheorem{lemma}{Lemma}

\def\noflash#1{\setbox0=\hbox{#1}\hbox to 1\wd0{\hfill}}

\begin{document}

\title{New Efficient Error-Free Multi-Valued Consensus \\with Byzantine Failures
 \footnote{\normalsize This research is supported
in part by Army Research Office grant W-911-NF-0710287 and National
Science Foundation award 1059540. Any opinions, findings, and conclusions or recommendations expressed here are those of the authors and do not
necessarily reflect the views of the funding agencies or the U.S. government.}}

\author{Guanfeng Liang and Nitin Vaidya\\ \normalsize Department of Electrical and Computer Engineering, and\\ \normalsize Coordinated Science Laboratory\\ \normalsize University of Illinois at Urbana-Champaign\\ \normalsize gliang2@illinois.edu, nhv@illinois.edu}

\maketitle
\date{}
%\input{abstract}

%\noindent
%Guanfeng Liang is a full time Ph.D student in University of Illinois at Urbana-Champaign. This paper should be considered for the best student paper award.

\thispagestyle{empty}

\newpage

\setcounter{page}{1}
In this report, we investigate the multi-valued Byzantine consensus problem described as follows: There are $n$ processors, namely $P_1,...,P_n$, of which at most $t$ processors may be {\em faulty} and deviate from the algorithm in arbitrary fashion. Denote the set of all fault-free processors as $P_{good}$.
Each processor $P_i$ is given an $L$-bit input value $v_i$, and they want
to agree on a value $v'$ such that the following properties are satisfied:
\begin{itemize}
\item {\em Termination}: every fault-free $P_i$ eventually decides on an output value $v_i'$,
\item {\em Consistency}: the output values of all fault-free processors are equal, i.e.,
for every fault-free processor $P_i$, $v_i'=v'$ for some $v'$,
\item {\em Validity}: if every fault-free $P_i$ holds the same input $v_i=v$ for some $v$, then $v'=v$.
\end{itemize}
Algorithms that satisfy the above properties in all executions
are said to be {\bf error-free}.

The discussion in this report is not self-contained, and relies heavily on the material in \cite{podc2011_consensus} and \cite{techreport_BA_complexity} -- please refer to these papers for necessary background.

\section{A More Efficient Consensus Algorithm}
In our recent paper \cite{podc2011_consensus} we introduced an  algorithm that solves this problem error-free with communication complexity approximately $\frac{n(n-1)}{n-2t}L$, for large enough $L$. In this report, we are going to present a more efficient algorithm. The consensus algorithm in this report achieves communication complexity 
\begin{equation}
\frac{n(n-1)}{n-t}L~bits
\end{equation}
for $t<n/3$ and sufficiently large $L$.

\begin{algorithm*}[!ht]
\caption{Multi-Valued Consensus (generation $g$)}\label{alg:consensus}
\begin{enumerate*}
\item {\bf Matching Stage: }\\
In the following steps, for every processor $P_i$: $R_i[k]\leftarrow$  $S_j[k]$ whenever $P_i$ receives $S_j[k]$ from its trusted processor $P_j$.

Each processor $P_i\in P_{match}$ performs steps 1(a) and 1(b) as follows:
\begin{enumerate*}
\item Compute $(S_{i}[1],\dots,S_{i}[n]) = C_{n-t}(v_i(g))$, and {\em send} $S_{i}[i]$ to every trusted processor $P_j$.
(including those not in $P_{match}$, and $P_i$ itself.).

\item $\forall P_j$ that trusts $P_i$:\\
 If  $P_i=\min\{l|P_l\in P_{match} ~\mathrm{ and }~ P_j ~\mathrm{ trusts }~ P_l\}$, then 
$P_i$ sends $S_i[k]$ to $P_j$ for each $k$ such that $P_j$ does not trust $P_k\in P_{match}$.
\end{enumerate*}

Each processor $P_j\notin P_{match}$ performs step 1(c) as follows:
\begin{enumerate*}
\setcounter{enumii}{2}
\item Using the first $n-t$ symbols it has received in steps 1(a) and 1(b), $P_j$ computes $S_j[j]$ according to $C_{n-t}$, then sends $S_j[j]$ to all trusted processors (Including $P_j$ itself.).
\end{enumerate*}

\item {\bf Checking Stage:}\\
Each processor $P_i$ (in $P_{match}$ or not) performs Checking Stage as follows:
\begin{enumerate*}
\item If $R_{i}\in C_{n-t}$ then $Detected_i \leftarrow$ \FALSE; else  $Detected_i \leftarrow$ \TRUE.

\item If $P_i\in P_{match}$ and $R_{i} \neq S_{i}$ then $Detected_i\leftarrow$ \TRUE.
\item Broadcast $Detected_i$ using \BSB.

\item \label{step:no_detect}
Receive $Detected_j$ from each processor $P_j$ (broadcast in step 2(c)).\\
If $Detected_j=$ \FALSE for all $P_j$, decide on 
$v_i'(g) = C^{-1}_{n-t}(R_{i})$; else enter Diagnosis Stage.
\end{enumerate*}

\item {\bf Diagnosis Stage: }\\
Each processor $P_i$ (in $P_{match}$ or not) performs Diagnosis Stage as follows:
\begin{enumerate*}
\item Broadcast $S_i$ and $R_i$ using \BSB.
\item $S^\#_j\leftarrow S_j$ and $R^\#_j\leftarrow R_j$ received from $P_j$ as a result of broadcast in step 3(a).
\end{enumerate*}

Using the broadcast information, all processors perform the following steps identically:
\begin{enumerate*}
\setcounter{enumii}{2}
\item For each edge  $(i,j)$ in \DG:
Remove edge $(i,j)$ if $\exists k$, such that $P_j$ receives $S_i[k]$ from $P_i$ in Matching stage and $R^\#_j[k] \neq S^\#_i[k]$
\item For each $P_i\in P_{match}$:
If $S^\#_i\notin C_{n-t}$, then  $P_i$ must be faulty. So remove $i$ and the adjacent edges from \DG.
\item For each $P_j\notin P_{match}$:
If $S^\#_j[j]$ is not consistent with the subset of $n-t$ symbols of $R^\#_j$, from which $S^\#_j[j]$ is computed, $P_j$ must be faulty. So remove $j$ and the adjacent edges from \DG.
\item If at least $t+1$ edges at any vertex $i$ have been removed, then  $P_i$ must be faulty. So remove $i$ and the adjacent edges.

%\item $v^\#_i\leftarrow C^{-1}(S^\#_i)$ for every processor $P_i\in P_{match}$.

\item Find the maximum set of processors $P_{new} \subseteq P_{match}$  such that $S^\#_i = S^\#_j$ for every pair of $P_i,P_j\in P_{new}$. In case of a tie, pick any one.
\item If $|P_{new}| < n-t$, terminate the algorithm and decide on the default output.\\
Else, decide on $v'_i(g) = C_{n-t}^{-1}(S^\#_j)$ for any $P_j\in P_{new}$, and update $P_{match} = P_{new}$.
\end{enumerate*}
\end{enumerate*}
\end{algorithm*}

Our algorithm achieves consensus on a long value of $L$ bits deterministically. Similar to the algorithm in \cite{podc2011_consensus}, the proposed 
 algorithm  progresses in generations. Each processor $P_i$ is given an input value $v_i$ of $L$ bits, which is divided into $L/D$ parts 
of size $D$ bits each. These parts are denoted as $v_i(1), v_i(2), \cdots , v_i(L/D)$. For the $g$-th generation ($1\le g \le L/D$),
each processor $P_i$ uses $v_i(g)$ as its input in Algorithm \ref{alg:consensus}. Each generation of the algorithm results in processor $P_i$ deciding on $g$-th part (namely, $v_i'(g)$) of its final decision value $v_i'$.

The value $v_i(g)$ is represented by a vector of $n-t$ symbols,
each symbol represented with $D/(n-t)$ bits. For convenience of presentation,
we assume that $D/(n-t)$ is an integer.
We will refer to these $n-t$ symbols as the {\em data symbols}.

An $(n,n-t)$ distance-$(t+1)$ Reed-Solomon code, denoted as $C_{n-t}$, is used to encode the
$n-t$ data symbols into $n$ {\em coded symbols}.
We assume that $D/(n-t)$ is large enough to allow the above Reed-Solomon code to exist,
specifically, $n \le 2^{D/(n-t)} - 1$. This condition is met only if $L$ is large enough (since $L > D$).

In each generation $g$, a set of at least $n-t$ processors that appear to have identical inputs up to generation $g-1$ is maintained. More formally, our algorithm maintain a set $P_{match}$ of size at least $n-t$ such that for every $P_i,P_j\in P_{match}$, $v_i(h) = v_j(h)$ appears to be true for all $h<g$. $P_{match}$ is updated in every generation. Notice that, in a particular generation, if $P_{match}$ does not exist, i.e., there are at least $t+1$ processors that appear to have input values different from the other processors, it can be guarantee that the fault-free nodes do not have identical inputs. Then our algorithm will terminate and all fault-free nodes will decide on a default output.

Initially (generation 1), $P_{match}$ is the set of all $n$ processors. The  operations in each generation $g$ are presented in Algorithm \ref{alg:consensus}

\subsection{Proof of Correctness}
In this section, we prove the correctness of Algorithm \ref{alg:consensus}. In the proofs of the following lemmas, we assume that the fault-free processors always trust each other \cite{podc2011_consensus}.
\begin{lemma}\label{lm:detect}
If $Detected_j=$\FALSE for all $P_j$ in Line 2(d), all fault-free processors $P_i \in P_{good}$ decide on the identical output value $v'(g)$ such that $v'(g)=v_j(g)$ for all $P_j\in P_{good}\cap P_{match}$.
\end{lemma}
\begin{proof}
According to the algorithm, every fault-free processor $P_i\in P_{good}$ has sent $S_i[i]$ (computed from $v_i(g)$ directly if $P_i\in P_{match}$, or computed using symbols received in Lines 1(a) and 1(b)  if $P_i \notin P_{match}$) to all the other fault-free processors. As a result, $R_i|P_{good} = R_j|P_{good}$ is true for every pair of fault-free processors $P_i,P_j\in P_{good}$. Since $|P_{good}|\ge n-t$ and $C_{n-t}$ is a distance-($t+1$) code, it follows that either all fault-free processors $P_{good}$ decide on the same output, or at least one fault-free processor $P_i\in P_{good}$ sets $Detected_i\leftarrow$\TRUE in Line 2(a). In the case all $Detected_j$=\FALSE, all fault-free processors decide on an identical $v'(g)$. Moreover, according to Line 2(b), every fault-free processor $P_j\in P_{good}\cap P_{match}$ finds $R_j = S_j$. It then follows that $v'(g) = C_{n-t}^{-t}(R_j) =C_{n-t}^{-t}(S_j)= v_j(g)$.

\end{proof}

\begin{lemma}\label{lm:consistency}
If a $P_{new}$ such that $|P_{new}|\ge n-t$ is found in Line 3(g), all fault-free processors $P_i \in P_{good}$ decide on the identical output value $v'(g)$ such that $v'(g)=v_j(g)$ for all $P_j\in P_{good}\cap P_{new}$.
\end{lemma}
\begin{proof}
Since $|P_{new}|\ge n-t$ and since at most $t$  processors are faulty, there must be at least $n-2t$ fault-free processors in $P_{good}\cap P_{new}$, which have broadcast the same $S^\#$'s in Line 3(b). So at Line 3(h), all fault-free processors decide on the identical output $v'(g)=v_j(g)$ for all $P_j\in P_{good}\cap P_{new}$.
\end{proof}

\begin{lemma}\label{lm:default}
If a $P_{new}$ such that $|P_{new}|\ge n-t$ can not be found in Line 3(g), then there must be two fault-free processors $P_i,P_j\in P_{good}$ such that $v_i\neq v_j$.
\end{lemma}
\begin{proof}
It is easy to see that if all fault-free processors in $P_{good}$ are given the same input, then a $P_{new}$ such that $|P_{new}|\ge n-t$ can always be found in Line 3(g). Then the lemma follows.
\end{proof}

For the correctness of the way \DG is updated, please see \cite{techreport_BA_complexity} and \cite{podc2011_consensus}. Now we can conclude the correctness of Algorithm \ref{alg:consensus} as the following theorem:

\begin{theorem}\label{thm:correct}
Given $n$ processors with at most $t<n/3$ are faulty, each given an input value of $L$ bits, Algorithm \ref{alg:consensus} achieves consensus correctly in $L/D$ generations, with the diagnosis stage performed for at most $t+t(t+1)$ times.
\end{theorem}
\begin{proof}
According to Lemmas \ref{lm:detect} and \ref{lm:consistency}, the decided output $v'(g)$ always equals to $v_j$ for some $P_j\in P_{good}\cap P_{match}$, unless $|P_{new}|<n-t$ in Line 3(h). So consistency and validity properties are satisfied until $|P_{new}|$ becomes $<n-t$. In the case $|P_{new}|<n-t$, according to Lemma \ref{lm:default}, there must be two fault-free processors that are given different inputs. Then it is safe to decide on a default output and terminate. So the $L$-bit output satisfies the consistency and validity properties.

Every time the diagnosis stage is performed, either at least one edge associated with a faulty processor is removed, or at least one processor is removed from $P_{match}$. So it takes at most $t(t+1)$ instances of the diagnosis stage before all faulty processors are identified. In addition, it will take at most $t$ instances to remove fault-free processors from $P_{match}$ until two fault-free processors are identified as having different inputs, and the algorithm terminates with a default output.
\end{proof}

\subsection{Complexity}
According to Theorem \ref{thm:correct}, we can compute  the communication complexity of Algorithm \ref{alg:consensus} in a  similar way as in \cite{techreport_BA_complexity} and \cite{podc2011_consensus}. With a appropriate choice of $D$, the complexity of Algorithm \ref{alg:consensus} can be made equal to 
\begin{equation}
\frac{n(n-1)}{n-t}L + O(n^4L^{0.5}).
\end{equation}
So for sufficiently large $L$ ($\Omega(n^6)$), the complexity is $O\left(nL\right)$.

\section{More Efficient $q$-validity Consensus}
In \cite{podc2011_consensus}, we also introduced an algorithm that solves consensus while satisfying the ``$q$-validity'' property, as stated below, for all $t+1\le q\le n-t$ with communication complexity $\frac{n(n-1)}{q-t}L$.

\begin{itemize}
\item $q$-Validity: If at least $q$ fault-free processors hold an identical input $v$, then the output $v'$ agreed by the fault-free processors equals input $v_j$ for some fault-free processor $P_j$. Furthermore, if $q\ge \lceil\frac{ n+1}{ 2}\rceil$, then $v'= v$.
\end{itemize}
When $q= t+1$, its complexity becomes $n(n-1)L$, which is not linear in $n$ any more. In fact, this algorithm  achieves communication complexity $O(nL)$ only when $q-t = \Omega(n)$. 

On the other hand, Algorithm \ref{alg:consensus} can achieve $q$-validity for $q\ge \lceil\frac{n+1}{2}\rceil$ with communication complexity $\frac{n(n-1)}{q}L$, if we substitute every ``$n-t$'' with ``$q$'' in the algorithm. This formulation of complexity is independent of $t$, and remains to be $O(n)$ as long as $q=\Omega(n)$. However, Algorithm \ref{alg:consensus} with the mentioned modification cannot achieve $q$-validity for any $q<\lceil\frac{n+1}{2}\rceil$.

In this section, we present an algorithm that achieves $q$-validity for all $t+1\le q\le n-t$ while keeping the complexity $O(nL)$, as long as  $q=\Omega(n)$. This algorithm uses the ``clique formation'' technique from our previous algorithm in \cite{podc2011_consensus} to achieve $q$-validity when $q$ is small, and uses the technique from Algorithm \ref{alg:consensus} presented in the previous section to improve communication complexity.

The value $v_i(g)$ is represented by a vector of $q$ data symbols,
each symbol represented with $D/q$ bits. 
An $(n,q)$ distance-$(n-q+1)$ Reed-Solomon code, denoted as $C_{q}$, is used to encode the
$q$ data symbols into $n$ coded symbols. The  operations in each generation $g$ are presented in Algorithm \ref{alg:q-validity}

%In each generation $g$, the algorithm tries to identify a set $P_{match}$ of  $q$ nodes that appear to have identical inputs in this generation. 

\begin{algorithm*}[!ht]
\caption{$q$-Validity Consensus, Matching and Checking stages (generation $g$)}\label{alg:q-validity}

\begin{enumerate*}
\item {\bf Matching Stage: }\\
In the following steps, for every processor $P_i$: $R_i[k]\leftarrow$  $S_j[k]$ whenever $P_i$ receives $S_j[k]$ from its trusted processor $P_j$.

Every processor $P_i$ performs steps 1(a) to 1(e) as follows:
\begin{enumerate*}
\item Compute $(S_{i}[1],\dots,S_{i}[n]) = C_{q}(v_i(g))$, and {\em send} $S_{i}[i]$ to every trusted processor $P_j$.

%\item \label{step:send_S}
%$
%R_{i}[j] \leftarrow \left\{ 
%\begin{array}{l}
%\textrm{symbol that $P_i$ receives from $P_j$, if $P_i$ trusts $P_j$;}\\
%\perp, \textrm{otherwise}
%\end{array}
%\right.
%$

\item \label{step:M}
If $S_{i}[j] = R_{i}[j]$ then $M_{i}[j] \leftarrow $ \TRUE; else $M_{i}[j] \leftarrow $ \FALSE
\item $P_i$ broadcasts the vector $M_i$ using \BSB
\end{enumerate*}
Using the received $M$ vectors:
\begin{enumerate*}
\setcounter{enumii}{3}
\item Find a set of processors $P_{match}$ of size $q$ such that \\
	\hspace*{0.3in} $M_{j}[k]=M_{k}[j]=$ \TRUE for every pair of $P_j,P_k\in P_{match}$. If multiple possibility exist for $P_{match}$, then any one of the possible
  sets is chosen arbitrarily as $P_{match}$ (all fault-free nodes choose a deterministic algorithm to select identical $P_{match}$).
\item If $P_{match}$ does not exist, then decide on a default value and continue to the next generation;\\ else continue to the following steps.

~

{\bf Note:} At this point, if $P_{match}$ does not exist, it is, in fact, safe to terminate the algorithm with a default output since it can be asserted that no $q$ fault-free nodes have identical inputs. However, by continuing to the next generation instead of terminating, $q$-validity is satisfied for the inputs of each individual generation.

\end{enumerate*}

When $P_{match}$ of size $q$ is found, each processor $P_i\in P_{match}$ performs step 1(g) as follows:
\begin{enumerate*}
\setcounter{enumii}{5}
\item $\forall P_j$ that trusts $P_i$:\\
 If  $i=\min\{l|P_l\in P_{match} ~\mathrm{ and }~ P_j ~\mathrm{ trusts }~ P_l\}$, then 
$P_i$ sends $S_i[k]$ to $P_j$ for each $k$ such that $P_j$ does not trust $P_k$.
\end{enumerate*}

Each processor $P_j\notin P_{match}$ performs step 1(g) as follows:
\begin{enumerate*}
\setcounter{enumii}{6}
\item Using the first $q$ symbols it has received from the processors in $P_{match}$ in steps 1(a) and 1(f), $P_j$ computes $S_j[j]$ according to $C_{q}$, then sends $S_j[j]$ to all trusted processors.

~

{\bf Note:} For every processor $P_i$ trusted by $P_j$, it has set $R_i[j]$ to the $S_j[j]$ received from $P_j$ in step 1(a). It will be replaced with the new $S_j[j]$ received in step 1(g).
\end{enumerate*}

\item {\bf Checking Stage:}\\
Each processor $P_i$ (in $P_{match}$ or not) performs Checking Stage as follows:
\begin{enumerate*}
\item If $R_{i}\in C_{q}$ then $Detected_i \leftarrow$ \FALSE; else  $Detected_i \leftarrow$ \TRUE.

\item If $P_i\in P_{match}$ and $R_{i} \neq S_{i}$ then $Detected_i\leftarrow$ \TRUE.
\item Broadcast $Detected_i$ using \BSB.

\item 
Receive $Detected_j$ from each processor $P_j$ (broadcast in step 2(c)).\\
If $Detected_j=$ \FALSE for all $P_j$, then decide on 
$v_i'(g) = C^{-1}_{q}(R_{i})$; else enter Diagnosis Stage
\end{enumerate*}
\end{enumerate*}
\end{algorithm*}

\begin{algorithm*}
\setcounter{algorithm}{1}
\caption{$q$-Validity Consensus, Diagnosis stage (generation $g$)}\label{alg:q-validity2}
\begin{enumerate*}
\setcounter{enumi}{2}
\item {\bf Diagnosis Stage: }\\
Each processor $P_i$ (in $P_{match}$ or not) performs Diagnosis Stage as follows:
\begin{enumerate*}
\item Broadcast $S_i$ and $R_i$ using \BSB.
\item $S^\#_j\leftarrow S_j$ and $R^\#_j\leftarrow R_j$ received from $P_j$ as a result of broadcast in step 3(a).
\end{enumerate*}

Using the broadcast information, all processors perform the following steps identically:
\begin{enumerate*}
\setcounter{enumii}{2}
\item For each edge  $(i,j)$ in \DG:
Remove edge $(i,j)$ if $\exists k$, such that $P_j$ receives $S_i[k]$ from $P_i$ in Matching stage and $R^\#_j[k] \neq S^\#_i[k]$.
\item For each $P_i\in P_{match}$:
If $S^\#_i\notin C_{q}$, then  $P_i$ must be faulty. So remove $i$ and the adjacent edges from \DG.
\item For each $P_j\notin P_{match}$:
If $S^\#_j[j]$ is not consistent with the subset of $q$ symbols of $R^\#_j|P_{match}$, from which $S^\#_j[j]$ is computed, $P_j$ must be faulty. So remove $j$ and the adjacent edges from \DG.
\item If at least $t+1$ edges at any vertex $i$ have been removed, then  $P_i$ must be faulty. So remove $i$ and the adjacent edges.

%\item $v^\#_i\leftarrow C^{-1}(S^\#_i)$ for every processor $P_i\in P_{match}$.

\item Find a set of processors $P_{decide} \subseteq P$  such that $S^\#_i = S^\#_j$ for every pair of $P_i,P_j\in P_{decide}$. In case of a tie, pick any one.
\item If $|P_{decide}| < q$, decide on the default output.\\
Else, decide on $v'_i(g) = C_{q}^{-1}(S^\#_j)$ for any $P_j\in P_{decide}$.
\end{enumerate*}
\end{enumerate*}
\end{algorithm*}

\subsection{Proof of Correctness}
\begin{lemma}
If there are a set of at least $q$ fault-free processors $Q\subseteq P_{good}$ such that for each $P_i\in Q$, $v_i(g) = v(g)$ for some $v(g)$, then a set $P_{match}$ of size $q$ necessarily exists. 
\end{lemma}

\begin{proof}
Since all the fault-free processors in $Q$ have identical input $v(g)$, $S_i = C_{q}(v(g))$ for all $P_i\in Q$. Since these processors are fault-free and always trust each other, they send each other correct messages in the matching stage. Thus, $R_i[j]=S_j[j]=S_i[j]$ for all $P_i,P_j\in Q$. This fact implies that $M_i[j]=M_j[i]=$\TRUE for all $P_i,P_j\in Q$. Since there are $|Q|\ge q$ fault-free processors in $Q$, it follows that a set $P_{match}$ of size $q$ must exist.
\end{proof}

\comment{
\begin{lemma}
All processors in $P_{match}\cap P_{good}$ have identical input in generation $g$.
\end{lemma}
\begin{proof}
$|P_{match}\cap P_{good}|\ge |P_{match}|-t = q-t$ since there are at most $t$ faulty processors. Consider any two processors $P_i,P_j\in P_{match}\cap P_{good}$. Since $M_i[j]=M_j[i] =$\TRUE, it follows that $S_i[i]=S_j[i]$ and $S_j[j]=S_i[j]$. Since there are $q-t$ fault-free processors in $P_{match}\cap P_{good}$, this implies that the codewords computed by these fault-free processors (in Line 1(a)) contain at least $q-t$ identical symbols. Since the code $C_{q-t}$ has dimension $(q-t)$, this implies that the fault-free processors in $P_{match}\cap P_{good}$ must have identical input in generation $g$.
\end{proof}

}

\begin{lemma}
If $Detected_j=$\FALSE for all $P_j$ in Line 2(d), all fault-free processors $P_i\in P_{good}$ decide on the identical output value $v'(g)$ such that $v'(g) = v_j(g)$ for all $P_j\in P_{match}\cap P_{good}$.
\end{lemma}
\begin{proof}
Observe that size of set $P_{match}\cap P_{good}$ is at least $q-t\ge 1$, so there must be at least one fault-free processor in $P_{match}$.
% and hence the decoding operations $C_{q}^{-1}(R_i|P_{match})$ and $C_{q-t}^{-1}(R_i|P_{match}\cap P_{good})$ are both defined.

%Since fault-free processors send correct messages, for all fault-free processors $P_i\in P_{good}$, $R_i|P_{good}$ are identical. Since $Detected_j=$\FALSE for all $P_j$, every $P_i\in P_{good}$ decides on $C_{q}^{-1}(R_i)$ as its output. Also, $C_{q}^{-1}(R_i) = C_{q}^{-1}(R_i|P_{match})$, since $C_{q}$ has dimension $(q)$ and $|P_{match}|=q$. It then follows that all the fault-free processors $P_i$ decide on the identical value $v'(g) = C_{q}^{-1}(R_i|P_{match}$ in Line 2(c). Since $R_j|P_{match}\cap P_{good} = S_j|P_{match}\cap P_{good}$ for all processors $P_j\in P_{match}\cap P_{good}$, $v'(g) = v_j(g)$ for all $P_j\in P_{match}\cap P_{good}$.

According to the algorithm, every fault-free processor $P_i\in P_{good}$ has sent $S_i[i]$ (computed from $v_i(g)$ directly if $P_i\in P_{match}$, or computed using the $q$ symbols received from $P_{match}$ in Lines 1(a) and 1(f)  if $P_i \notin P_{match}$) to all the other fault-free processors. As a result, $R_i|P_{good} = R_j|P_{good}$ is true for every pair of fault-free processors $P_i,P_j\in P_{good}$. Since $|P_{good}|\ge n-t\ge q$ and $C_{q}$ has dimension $q$, it follows that either all fault-free processors $P_{good}$ decide on the same output, or at least one fault-free processor $P_i\in P_{good}$ sets $Detected_i\leftarrow$\TRUE in Line 2(a). In the case $Detected_j$=\FALSE for all $P_j$, all fault-free processors decide on an identical $v'(g)$. Moreover, according to Line 2(b), every fault-free processor $P_j\in P_{good}\cap P_{match}$ finds $R_j = S_j$. It then follows that $v'(g) = C_{q}^{-t}(R_j) =C_{q}^{-t}(S_j)= v_j(g)$ where $P_j \in P_{good}\cap P_{match}$.

\end{proof}

Then we can have the following theorem about the correctness of Algorithm \ref{alg:q-validity}.
\begin{theorem}
Given $n$ processors with at most $t<n/3$ are faulty, each given an input value of $L$ bits, Algorithm \ref{alg:q-validity} achieves $q$-validity for each one of the $L/D$ generations, with the diagnosis stage performed for at most $t(t+1)$ times.
\end{theorem}
\begin{proof}
Similar to Theorem \ref{thm:correct}.
\end{proof}

\subsection{Complexity}
In Lines 1(a) and 1(f), every processor receives at most $n-1$ symbols, so at most $n(n-1)$ symbols are communicated in these two steps. In Line 1(g), every processor $P_j\notin P_{match}$ sends at most $n-1$ symbols, and there are at most $n-q$ processors not in $P_{match}$, so at most $(n-q)(n-1)$ symbols are communicated in this step. So in total, no more than $(2n-q)(n-1)$ symbols are communicated in the Matching stage. Then with an appropriate choice of $D$, the complexity of Algorithm \ref{alg:q-validity} can be made to 
\begin{equation}
\le \frac{(2n-q)(n-1)}{q}L + O(n^4L^{0.5}).
\end{equation}
So for any $q=\Omega(n)$ and $t+1\le q\le n-t$, with a sufficiently large $L$ ($\Omega(n^6)$), the complexity is $O(nL)$.

%\newpage
%\bibliographystyle{abbrv}
\bibliography{../PaperList}

\begin{thebibliography}{1}

\bibitem{techreport_BA_complexity}
Guanfeng Liang and Nitin Vaidya.
\newblock Complexity of multi-valued byzantine agreement.
\newblock {\em Technical Report, CSL, UIUC (http://arxiv.org/abs/1006.2422)},
  June 2010.

\bibitem{podc2011_consensus}
Guanfeng Liang and Nitin Vaidya.
\newblock Error-free multi-valued consensus with byzantine failures.
\newblock In {\em ACM PODC}, 2011.

\end{thebibliography}

\end{document}